\documentclass[runningheads]{llncs}
\usepackage[in]{fullpage}

\newif\ifsubmission
\submissionfalse

\usepackage{amsthm}
\usepackage[T1]{fontenc}
\usepackage{graphicx}
\usepackage{amsmath}
\usepackage{amssymb,amsfonts}
\usepackage[
n,
advantage,
operators,
sets,
adversary,
landau,
probability,
notions,
logic,
ff,
mm,
primitives,
events,
complexity,
asymptotics,
keys,
]{cryptocode}
\usepackage[
colorlinks,
linkcolor = black,
citecolor = blue,
urlcolor = blue]{hyperref}%
\usepackage[nameinlink]{cleveref}
\usepackage{physics}

\newcommand{\lang}{\mathcal{L}}

\newcommand{\dom}{\{0,1\}^*}
\newcommand{\promiseBQP}{\hyperref[def:promiseBQP]{\mathsf{PromiseBQP}}}
\newcommand{\promiseQMA}{\hyperref[def:promiseQMA]{\mathsf{PromiseQMA}}}
\newcommand\pr[2][]{\Pr_{#1}\left[#2\right]}
\newcommand\given[1][]{\:#1\middle\rvert\:}

\createpseudocodeblock{pcb}{boxed,center}{}{}{} %

\begin{document}
\title{Quantum Pseudorandomness Cannot Be Shrunk\\ In a Black-Box Way}

\ifsubmission
\author{}
\institute{}
\else
\author{Samuel Bouaziz{-}{-}Ermann\inst{1} \and
  Garazi Muguruza\inst{2,3}}
\institute{Sorbonne Universit\'e, CNRS, LIP6, France \and
  Informatics Institute, University of Amsterdam, Netherlands \and
  QuSoft, Netherlands}
    \authorrunning{S. Bouaziz{-}{-}Ermann \and G. Muguruza}
\fi
\maketitle
\begin{abstract}
Pseudorandom Quantum States (PRS) were introduced by Ji, Liu and Song as quantum analogous to Pseudorandom Generators.
They are an ensemble of states efficiently computable but computationally indistinguishable from Haar random states.
Subsequent works have shown that some cryptographic primitives can be constructed from PRSs.
Moreover, recent classical and quantum oracle separations of PRS from One-Way Functions strengthen the interest in a purely quantum alternative building block for quantum cryptography, potentially weaker than OWFs.

However, our lack of knowledge of extending or shrinking the number of qubits of the PRS output still makes it difficult to reproduce some of the classical proof techniques and results.
Short-PRSs, that is PRSs with logarithmic size output, have been introduced in the literature along with cryptographic applications, but we still do not know how they relate to PRSs.
Here we answer half of the question, by showing that it is not possible to shrink the output of a PRS from polynomial to logarithmic qubit length while still preserving the pseudorandomness property, in a relativized way.
More precisely, we show that relative to Kretschmer's quantum oracle (TQC 2021) short-PRSs cannot exist (while PRSs exist, as shown by Kretschmer's work).
\end{abstract}
\section{Introduction}

Pseudorandomness is an important concept in cryptography since almost all relevant classical cryptographic primitives require the existence of one-way functions (OWF) or equivalent objects such as pseudorandom generators (PRG) and pseudorandom functions (PRF)~\cite{classicalPRNG}.
It corresponds to a deterministic function whose output cannot be distinguished from the uniform distribution by a computationally bounded algorithm. 

In 2018, Ji, Liu and Song~\cite{C:JiLiuSon18} introduced a quantum analog of PRGs, called \emph{Pseudorandom Quantum States} (PRS) that consists of a family of polynomial size keyed-states $\{\ket{\phi_k}\}_{k\in\mathcal{K}}$ such that no quantum polynomial-time algorithm can distinguish between a polynomial number of copies\footnote{The reason we give the adversary a polynomial number of copies of the state is that an arbitrary quantum state cannot be cloned, by the no-cloning theorem.} of a randomly sampled element from the PRSs family or a polynomial number of copies of a Haar-random state (see~\Cref{def:prs} for a formal definition). 
We already know how to construct several cryptographic primitives from (variants of) PRSs: public key encryption with quantum keys~\cite{TCC:BGHMSVW23}, quantum digital signatures~\cite{C:MorYam22}, pseudo one-time pad encryption schemes~\cite{C:AnaQiaYue22}, statistically binding and computationally hiding commitments~\cite{C:AnaQiaYue22,MY22b,KhuTom23} and quantum computational zero knowledge proofs~\cite{ITCS:BraCanQia23}. %
Such rapid interest derives probably from the fact that PRSs can be constructed from OWFs~\cite{C:JiLiuSon18} (and thus PRGs), but there are oracle separations found between OWFs and PRSs~\cite{Kre21,STOC23}, which makes them a potentially weaker building block for quantum cryptography, with a purely quantum description.

Classically, practical applications of PRGs require stretching the output, which can be done arbitrarily by just using the output of the PRG as an input and roughly composing the PRG with itself. However, no analogous intuitive construction is possible with PRSs because the input is a classical bit string and the output is a quantum state.
Moreover, although classically shrinking a PRG does not make sense, it is still possible by discarding part of the output of the original PRG.
Quantumly, discarding half of a quantum state could lead to a maximally mixed state, making it easy to distinguish from a Haar-random state, which means that pseudorandomness is not a property respected by subsets of the registers.

This motivates the definition of short-PRS, a PRS that on input $k\in\{0,1\}^\lambda$ outputs a qubit of size $n(\lambda)=O(\log\lambda)$.
Brakerski and Shmueli \cite{C:BraShm20} proved that $c\cdot\log\lambda$-output PRSs exist for any $c\geq1$.
While classically, ``short-PRGs'' and PRGs are equivalent, in the quantum setting, the following question remains open\\
\vspace{0.1cm}\\
\makebox[\textwidth]{Question 1: \emph{What is the relation between short and long PRSs?}}\\
\vspace{0.1cm}

In this work we answer half of the question, by showing that there exists a quantum oracle relative to which PRSs exist but short-PRSs do not, thus PRSs are a weaker assumption than short-PRSs.
Note that this question appeared in the literature~\cite{C:BraShm20,TCC:BitBra21,C:AnaQiaYue22}.

Our separation relies on another fundamental difference between quantum and classical cryptography; while pseudorandomness has received much attention in classical cryptography, the same cannot be said about \emph{pseudodeterminism}.
A pseudodeterministic variant of PRGs (PD-PRG) was first introduced by Ananth, Lin and Yuen~\cite{PRNG}, defined as a quantum polynomial-time algorithm that outputs a pseudorandom string on a \emph{fraction} of the input keys.
In their work, this fraction is polynomial in size and this define polynomial pseudodeterminism, and we define overwhelming pseudodeterminsm when this fraction is overwhelming in size.
They showed that short-PRSs can be used to construct PD-PRGs. %
Later Barhoush, Behera, Ozer, Salvail and Sattath~\cite{BS23} introduced the analogous pseudodeterministic one-way functions (PD-OWF) and showed that PD-PRGs imply PD-OWFs. 

\begin{figure}
    \centering
    \tikzset{every picture/.style={line width=0.75pt}} 

\begin{tikzpicture}[x=0.75pt,y=0.75pt,yscale=-1,xscale=1]

\draw (62.6,122.33) node    {$OWF$};
\draw (63.27,65.67) node    {$PD-OWF^{negl}$};
\draw (203.27,122) node    {$PRG$};
\draw (63.27,21.67) node    {$PD-OWF^{poly}$};
\draw (204.27,19.67) node    {$PD-PRG^{poly}$};
\draw (334.4,67.07) node    {$short-PRS$};
\draw (203.93,66) node    {$PD-PRG^{negl}$};
\draw (316.67,13.07) node [anchor=north west][inner sep=0.75pt]    {$PRS$};
\draw [color={rgb, 255:red, 0; green, 0; blue, 0 }  ,draw opacity=1 ]   (87.6,122.27) -- (180.27,122.05) ;
\draw [shift={(182.27,122.05)}, rotate = 179.86] [color={rgb, 255:red, 0; green, 0; blue, 0 }  ,draw opacity=1 ][line width=0.75]    (10.93,-3.29) .. controls (6.95,-1.4) and (3.31,-0.3) .. (0,0) .. controls (3.31,0.3) and (6.95,1.4) .. (10.93,3.29)   ;
\draw [shift={(85.6,122.28)}, rotate = 359.86] [color={rgb, 255:red, 0; green, 0; blue, 0 }  ,draw opacity=1 ][line width=0.75]    (10.93,-3.29) .. controls (6.95,-1.4) and (3.31,-0.3) .. (0,0) .. controls (3.31,0.3) and (6.95,1.4) .. (10.93,3.29)   ;
\draw [color={rgb, 255:red, 0; green, 0; blue, 0 }  ,draw opacity=1 ]   (62.74,110.33) -- (63.09,80.67) ;
\draw [shift={(63.11,78.67)}, rotate = 90.67] [color={rgb, 255:red, 0; green, 0; blue, 0 }  ,draw opacity=1 ][line width=0.75]    (10.93,-3.29) .. controls (6.95,-1.4) and (3.31,-0.3) .. (0,0) .. controls (3.31,0.3) and (6.95,1.4) .. (10.93,3.29)   ;
\draw [color={rgb, 255:red, 0; green, 0; blue, 0 }  ,draw opacity=1 ]   (119.77,60.8) -- (151.43,60.88) ;
\draw [shift={(117.77,60.8)}, rotate = 0.14] [color={rgb, 255:red, 0; green, 0; blue, 0 }  ,draw opacity=1 ][line width=0.75]    (10.93,-3.29) .. controls (6.95,-1.4) and (3.31,-0.3) .. (0,0) .. controls (3.31,0.3) and (6.95,1.4) .. (10.93,3.29)   ;
\draw [color={rgb, 255:red, 0; green, 0; blue, 0 }  ,draw opacity=1 ]   (204.03,53) -- (204.16,34.67) ;
\draw [shift={(204.17,32.67)}, rotate = 90.41] [color={rgb, 255:red, 0; green, 0; blue, 0 }  ,draw opacity=1 ][line width=0.75]    (10.93,-3.29) .. controls (6.95,-1.4) and (3.31,-0.3) .. (0,0) .. controls (3.31,0.3) and (6.95,1.4) .. (10.93,3.29)   ;
\draw [color={rgb, 255:red, 0; green, 0; blue, 0 }  ,draw opacity=1 ]   (63.27,52.67) -- (63.27,36.67) ;
\draw [shift={(63.27,34.67)}, rotate = 90] [color={rgb, 255:red, 0; green, 0; blue, 0 }  ,draw opacity=1 ][line width=0.75]    (10.93,-3.29) .. controls (6.95,-1.4) and (3.31,-0.3) .. (0,0) .. controls (3.31,0.3) and (6.95,1.4) .. (10.93,3.29)   ;
\draw [color={rgb, 255:red, 0; green, 0; blue, 0 }  ,draw opacity=1 ]   (119.77,15.86) -- (151.77,15.41) ;
\draw [shift={(117.77,15.89)}, rotate = 359.19] [color={rgb, 255:red, 0; green, 0; blue, 0 }  ,draw opacity=1 ][line width=0.75]    (10.93,-3.29) .. controls (6.95,-1.4) and (3.31,-0.3) .. (0,0) .. controls (3.31,0.3) and (6.95,1.4) .. (10.93,3.29)   ;
\draw [color={rgb, 255:red, 0; green, 0; blue, 0 }  ,draw opacity=1 ]   (241.84,33.35) -- (301.45,55.07) ;
\draw [shift={(239.96,32.67)}, rotate = 20.01] [color={rgb, 255:red, 0; green, 0; blue, 0 }  ,draw opacity=1 ][line width=0.75]    (10.93,-3.29) .. controls (6.95,-1.4) and (3.31,-0.3) .. (0,0) .. controls (3.31,0.3) and (6.95,1.4) .. (10.93,3.29)   ;
\draw [color={rgb, 255:red, 0; green, 0; blue, 0 }  ,draw opacity=1 ]   (203.41,110) -- (203.75,81) ;
\draw [shift={(203.78,79)}, rotate = 90.68] [color={rgb, 255:red, 0; green, 0; blue, 0 }  ,draw opacity=1 ][line width=0.75]    (10.93,-3.29) .. controls (6.95,-1.4) and (3.31,-0.3) .. (0,0) .. controls (3.31,0.3) and (6.95,1.4) .. (10.93,3.29)   ;
\draw [color={rgb, 255:red, 0; green, 0; blue, 0 }  ,draw opacity=1 ]   (316.85,79.84) -- (224.27,118.62) ;
\draw [shift={(318.7,79.07)}, rotate = 157.27] [color={rgb, 255:red, 0; green, 0; blue, 0 }  ,draw opacity=1 ][line width=0.75]    (10.93,-3.29) .. controls (6.95,-1.4) and (3.31,-0.3) .. (0,0) .. controls (3.31,0.3) and (6.95,1.4) .. (10.93,3.29)   ;
\draw [color={rgb, 255:red, 0; green, 0; blue, 0 }  ,draw opacity=1 ]   (317.78,32.67) -- (220.23,108.77) ;
\draw [shift={(218.65,110)}, rotate = 322.04] [color={rgb, 255:red, 0; green, 0; blue, 0 }  ,draw opacity=1 ][line width=0.75]    (10.93,-3.29) .. controls (6.95,-1.4) and (3.31,-0.3) .. (0,0) .. controls (3.31,0.3) and (6.95,1.4) .. (10.93,3.29)   ;
\draw [shift={(268.22,71.33)}, rotate = 187.04] [color={rgb, 255:red, 0; green, 0; blue, 0 }  ,draw opacity=1 ][line width=0.75]    (-5.59,0) -- (5.59,0)(0,5.59) -- (0,-5.59)   ;
\draw [color={rgb, 255:red, 245; green, 166; blue, 35 }  ,draw opacity=1 ]   (240.99,53) -- (311.78,28.17) ;
\draw [shift={(313.67,27.51)}, rotate = 160.67] [color={rgb, 255:red, 245; green, 166; blue, 35 }  ,draw opacity=1 ][line width=0.75]    (10.93,-3.29) .. controls (6.95,-1.4) and (3.31,-0.3) .. (0,0) .. controls (3.31,0.3) and (6.95,1.4) .. (10.93,3.29)   ;
\draw [color=blue  ,draw opacity=1 ]   (313.67,20.52) -- (258.77,20.09) ;
\draw [shift={(256.77,20.07)}, rotate = 0.44] [color=blue  ,draw opacity=1 ][line width=0.75]    (10.93,-3.29) .. controls (6.95,-1.4) and (3.31,-0.3) .. (0,0) .. controls (3.31,0.3) and (6.95,1.4) .. (10.93,3.29)   ;
\draw [shift={(285.22,20.29)}, rotate = 225.44] [color=blue  ,draw opacity=1 ][line width=0.75]    (-5.59,0) -- (5.59,0)(0,5.59) -- (0,-5.59)   ;
\draw [color={rgb, 255:red, 245; green, 166; blue, 35 }  ,draw opacity=1 ]   (117.77,25.89) -- (149.77,25.44) ;
\draw [shift={(151.77,25.41)}, rotate = 179.19] [color={rgb, 255:red, 245; green, 166; blue, 35 }  ,draw opacity=1 ][line width=0.75]    (10.93,-3.29) .. controls (6.95,-1.4) and (3.31,-0.3) .. (0,0) .. controls (3.31,0.3) and (6.95,1.4) .. (10.93,3.29)   ;
\draw [color={rgb, 255:red, 245; green, 166; blue, 35 }  ,draw opacity=1 ]   (117.77,70.8) -- (149.43,70.87) ;
\draw [shift={(151.43,70.88)}, rotate = 180.14] [color={rgb, 255:red, 245; green, 166; blue, 35 }  ,draw opacity=1 ][line width=0.75]    (10.93,-3.29) .. controls (6.95,-1.4) and (3.31,-0.3) .. (0,0) .. controls (3.31,0.3) and (6.95,1.4) .. (10.93,3.29)   ;
\draw [color={rgb, 255:red, 245; green, 166; blue, 35 }  ,draw opacity=1 ]   (292.81,79.07) -- (220.82,109.23) ;
\draw [shift={(218.97,110)}, rotate = 337.27] [color={rgb, 255:red, 245; green, 166; blue, 35 }  ,draw opacity=1 ][line width=0.75]    (10.93,-3.29) .. controls (6.95,-1.4) and (3.31,-0.3) .. (0,0) .. controls (3.31,0.3) and (6.95,1.4) .. (10.93,3.29)   ;
\draw [shift={(255.89,94.53)}, rotate = 202.27] [color={rgb, 255:red, 245; green, 166; blue, 35 }  ,draw opacity=1 ][line width=0.75]    (-5.59,0) -- (5.59,0)(0,5.59) -- (0,-5.59)   ;

\end{tikzpicture}
    \caption{Relation between different primitives. Proven, \textcolor{orange}{expected} and \textcolor{blue}{our result}.}\label{fig:graph}
\end{figure}
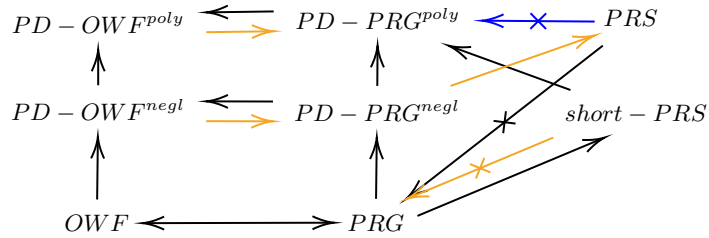

Although a quantum state can encode arbitrary classical information, this information is not necessarily \emph{accessible} for an observer. For example, while OWFs are separated from PRSs, PD-OWFs can be built from short-PRSs~\cite{BS23}. %
Another interesting question is thus\\
\vspace{0.1cm}\\
\makebox[\textwidth]{Question 2: \emph{What is the role of pseudodeterminism in quantum cryptography?%
}}\\
\vspace{0.1cm}

Pseudodeterministic variants of OWFs and PRGs seem to be enough to build many interesting tasks in cryptography. %
    For example,~\cite{PRNG} state that overwhelming PD-PRGs are enough to build PRSs, as these can be used to generate pseudorandom phases pseudodeterministically, and they build polynomial PD-PRGs from short-PRSs. However, it is unclear how to build overwhelming PD-PRGs from short-PRSs, or if polynomial PD-PRGs are enough to build PRSs\footnote{The same construction from PRGs does not work as the polynomial error will induce too much error in the resulting state to be a PRS.}.
In other words, there seems to be a gap between the cryptography we can build from polynomial or overwhelming pseudodeterminism, with short-PRSs being nearly as strong as overwhelming pseudodeterminism while PRSs are weaker than polynomial pseudodeterminism, as we show in this work.
What is clear is that pseudodeterminism is a non-trivial property in quantum cryptography.

\paragraph{Our contribution.} 

We show that Kretschmer's oracle~\cite{Kre21} not only implies that OWFs do not exist, but also none of the pseudoterministic variants do either.
Since PD-PRGs and PD-OWFs can be constructed from short-PRSs~\cite{PRNG,BS23}, our work gives a separation between PRSs and short-PRSs\footnote{Note that relative to Kretschmer's oracle, not only do we have PRSs, but we also have pseudorandom unitaries (PRU).} and can be stated as follows.
\begin{theorem}[\Cref{thm:main}, informal]
  \label{thm:maininformal}
  There exists a quantum oracle \(\mathcal{O}\) relative to which PRSs exist but short-PRSs do not.
\end{theorem}
This result might sound counterintuitive, as it shows that we cannot shrink a pseudorandom quantum state to a smaller one.
An explanation of this result could be that requiring a polynomial number of copies of a logarithmic quantum state to be indistinguishable from a polynomial number of copies of a random state is a strong assumption, that is similar to OWFs, as shown by previous works.

Relative to Kretschmer's oracle we know that poly-size PRSs exist and $\promiseBQP=\promiseQMA$.
Here we show that not only the existence of OWFs, but also the existence of polynomial error pseudodeterministic OWFs (a possibly weaker assumption, but implied by short-PRSs) also implies $\promiseBQP\not=\promiseQMA$. 
For the proof to work, we rely on the promise version of the complexity classes.
Promises problems are such that there are yes instances and no instances, but also other instances where the output of an algorithm does not matter.
In our proof, we define a language with the yes instances as the values for which there exists a \emph{high} probability pre-image, and the no instances are values for which there is no \emph{low} probability pre-image.
Thus there is a gap between the possible success probabilities, and this gap is needed to distinguish between the yes and no instances in polynomial time.

\paragraph{Open questions.}
We still leave open the question in the opposite direction, of building PRSs from short-PRSs.
There are two partial results regarding this in the literature; the first is that it is possible to construct PRSs with a bigger dimension, but at the cost of the security definition to hold only for single-copy PRSs~\cite{STOC23b}, the second gives a positive answer by assuming the existence of pseudorandom isometries~\cite{isometries}.
However, the general question, without assumptions, still remains open.

Note that our separation is relative to a quantum oracle, thus a \emph{classical} oracle separation of short-PRSs from PRSs is unclear.
The classical oracle from~\cite{STOC23} gives P=NP, but this is not enough for PD-OWFs, as we need promise problems in our proof.

As mentioned above, pseudodeterminism is a non-trivial feature of quantum cryptography. It is unclear where PRSs and short-PRSs lie in the graph of~\Cref{fig:graph}: are short-PRSs equivalent to $\negl$ error pseudodeterminism? Which cryptographic tasks require negligible error pseudodeterminism and which ones inverse polynomial? Moreover, although intuitively the answer should be yes, it is still unclear if all the pseudodeterministic variants of PRGs/PRFs/OWFs are equivalent to each other.
Constructing PD-OWFs from PD-PRGs is possible by the generalization lemma of~\cite{BS23}, but proving the other direction is still open.%

Finally, our results imply that the quantum oracle separation from~\cite{Kre21} is not enough to separate short-PRSs from OWFs, thus it is natural to ask if this is possible.
Indeed quantum oracles for ``short'' PRUs will not lead to $\promiseBQP=\promiseQMA$, because the concentration inequality on the Haar measure does not hold with small dimensions for the unitaries.
As for classical oracles, there is a classical oracle separation between 1-copy short-PRSs and OWFs~\cite{STOC23}, but a multi-copy separation is still unknown (both for short-PRSs and PRSs).
\section{Preliminaries}

We use $\lambda$ to denote the security parameter.
We use $\negl[\cdot]$ to denote a negligible function.
We use \(\epsilon\) to denote the empty string.
We use \(||\) to denote the concatenation operator.
We use \(x\prec y\) to denote the fact that \(x\) is a prefix of \(y\), i.e. there exists \(x'\) such that \(y = x || x'\).
We use \(x\nprec y\) to denote the fact that \(x\) is not a prefix of \(y\).
We use $\leftarrow\mathcal{A}$ to denote uniform sampling from the set $\mathcal{A}$.

We use $\mathcal{H}_n$ to denote the Haar measure over $n$-qubit space, i.e.~$\mathcal{H}((\mathbb{C})^{\otimes n})$. The Haar measure over $\mathbb{C}^d$ is the uniform measure over all $d$-dimensional unit vectors over $\mathbb{C}$.

We include here the relevant pseudorandom notions.
\begin{definition}[Pseudorandom quantum states \cite{C:JiLiuSon18}]
\label{def:prs}
Let $\lambda \in \mathbb{N}$, and let $n(\lambda)$ be the number of qubits in the quantum system.
A keyed family of $n$-qubit quantum states $\{\ket{\phi_k}\}_{k \in \{0,1\}^\lambda}$ is \emph{pseudorandom} if the following two conditions hold:
\begin{enumerate}
\item \textbf{Efficient generation}. There is a QPT algorithm $G$ that on input $k\in\{0,1\}^\lambda$ generates 
\begin{align*}
    G_\lambda(k)=\ketbra{\varphi_k}.
\end{align*}
\item \textbf{Pseudorandomness}. For any QPT adversary $\adv$ and all polynomials $t(\cdot)$, we have
\begin{align*}\left| \pr[k \leftarrow \{0,1\}^\lambda]{\adv\left(1^\lambda, \ket{\varphi_k}^{\otimes t(\lambda)}\right) = 1 } - \pr[\ket{\nu} \leftarrow \mathcal{H}_n(\lambda)]{\adv\left(1^\lambda, \ket{\nu}^{\otimes t(\lambda)}\right) = 1 } \right| \leq \negl[\lambda].\end{align*}
\end{enumerate}
\end{definition}

We say that a~$n(\lambda)$-PRS is a \emph{short-PRS} if the output is logarithmic in the security parameter, i.e.~$n(\lambda)=\Theta(\log\lambda)$. From now on we will use PRSs to refer to long-output PRSs and short-PRSs for logarithmic output.

We also include a pseudodeterministic primitive.
\begin{definition}[Quantum Pseudo-deterministic One-Way Functions\footnote{In~\cite{BS23} they actually define Quantum Pseudo-deterministic One-Way \emph{Hash} Functions (PD-OWHF). We omit the \emph{hash} property here for simplicity, but since the security properties of both functionalities are equivalent our proof also trivially works for PD-OWHF.}~{\cite[Definition 9]{BS23}}]
  \label{def:QPD-OWF}
    A QPT algorithm $F :\{0,1\}^{m(\lambda)}\rightarrow \{0,1\}^{\ell(\lambda)}$ is a \emph{quantum pseudo-deterministic one-way function} if the following conditions hold:
  \begin{itemize}
    \item \textbf{Pseudodeterminism}. There exists a constant $c>0$ and function $\mu(\lambda)=O(\lambda^{-c})$ such that for all $\lambda\in \mathbb{N}$, there exists a set $\mathcal{K}_\lambda \subset \{0,1\}^{m(\lambda)}$:\vspace{1mm}
      \begin{enumerate}
        \item $\pr{x\in\mathcal{K}_\lambda \given{x\leftarrow \{0,1\}^{\ell(\lambda)}}}\geq 1-\mu(\lambda)$.\vspace{1mm}
        \item For any $x\in \mathcal{K}_\lambda$, it holds that
          \begin{align}
            \label{eq:pd}\max_{y\in \{0,1\}^{\ell(\lambda)}}\pr{y=F_\lambda(x)}\geq 1-\negl[\lambda],
          \end{align}
          where the probability is over the randomness of $F_\lambda$.
      \end{enumerate}
    \item \textbf{Security}. For every QPT inverter $\adv$:
      \begin{align}
        \label{eq:sec} \pr[x\leftarrow \{0,1\}^{m(\lambda)}]{F\left(\adv(F(x))\right)=F(x)}\leq \negl[\lambda],
      \end{align}
      where the probability is over the randomness of $F$ and $\adv$.
  \end{itemize}
\end{definition}

Note that the pseudodeterminism factor in the above definition comes from the size of the \emph{good} key space~$\mu(\lambda)$, which is an inverse-polynomial in the security parameter~$\lambda$. This means that for a non-negligible number of elements in the key space, the OWF could behave arbitrarily.
We could also define a negligible variant by requesting $\mu(\lambda)$ to be a negligible function in $\lambda$.

We can build PD-OWFs from short-PRSs.
\begin{theorem}[Adapted from~{\cite[Theorem 6]{BS23}}\footnote{Here we also use the PD-OWF variant of their theorem originally for PD-OWHF. This choice affects the parameters of the domain and range in the theorem statement because constructing a PD-OWHF requires more steps than constructing a PD-OWF (we only need the first step of their proof). However, note that changing the domain/range of the function to some different polynomials in \(\lambda\) would still make the proof go through by changing some parameters in the proof.}]
  \label{thm:OWF}
  Assuming the existence of $(c\log \lambda)$-PRSs with $c>12$, there exists a $O(\lambda^{-c/12+1})$-PD-OWF~$F:\{0,1\}^{\ell(\lambda)}\rightarrow \{0,1\}^{\ell(\lambda)}$ with input/output length~$\ell(\lambda)=\lambda^{c/6}$.
\end{theorem}

Finally, we will need Kretschmer's (quantum) oracle~$\mathcal{O}$ relative to which OWFs do not exist, but PRSs do.
The former is because $\promiseBQP$ and $\promiseQMA$ are equal relative to this oracle%
, we only include here the definitions of these languages for the sake of self-containment, but we refer the reader to the original paper for a lengthier explanation of the complexity classes.
\begin{definition}
\label{def:promiseQMA}
A promise problem $\lang=\lang_{\mathrm{yes}}\cup\lang_{\mathrm{no}}\cup\lang_{\bot}$ with $\lang\subseteq \{0,1\}^*$ is in $\mathsf{PromiseQMA}$ (\underline{Q}uantum \underline{M}erlin--\underline{A}rthur) if there exists a polynomial-time quantum algorithm $\mathsf{V}(x,\ket{\psi})$ called a \emph{verifier} and a polynomial $p$ such that:
\begin{enumerate}
    \item (Completeness) If $x\in\lang_{\mathrm{yes}}$, then there exists a quantum state $\ket{\psi}$ on $p(|x|)$ qubits (called a \emph{witness} or \emph{proof}) such that $\Pr\left[\mathsf{V}(x,\ket{\psi}) = 1\right] \geq \frac{2}{3}$.
    \item (Soundness) If $x\in\lang_{\mathrm{no}}$, then for every quantum state $\ket{\psi}$ on $p(|x|)$ qubits, $\Pr\left[\mathsf{V}(x,\ket{\psi}) = 1\right] \leq \frac{1}{3}$.
\end{enumerate}
\end{definition}

\begin{definition}
\label{def:promiseBQP}
A promise problem $\lang=\lang_{\mathrm{yes}}\cup\lang_{\mathrm{no}}\cup\lang_{\bot}$ with $\lang\subseteq \{0,1\}^*$ is in $\mathsf{PromiseBQP}$ (\underline{B}ounded-error \underline{Q}uantum \underline{P}olynomial time) if there exists a randomized polynomial-time quantum algorithm $\adv(x)$ such that:
\begin{enumerate}
    \item If $x\in\lang_{\mathrm{yes}}$, then $\Pr\left[\adv(x) = 1\right] \geq \frac{2}{3}$.
    \item If $x\in\lang_{\mathrm{no}}$, then $\Pr\left[\adv(x) = 1\right] \leq \frac{1}{3}$.
\end{enumerate}
\end{definition}

\begin{theorem}[\cite{Kre21}]
\label{thm:kre21}
There exists a quantum oracle \(\mathcal{O}\), such that:
\begin{enumerate}
    \item $\promiseBQP^{\mathcal{O}} = \promiseQMA^{\mathcal{O}}$.
    \item $\lambda$-PRSs exist relative to $\mathcal{O}$.\footnote{The random oracle is a random unitary operation, which gives the existence of PRUs and thus PRSs~\cite{C:JiLiuSon18}.}
\end{enumerate}
\end{theorem}
\section{Separating PRSs from short-PRSs}

In this section, we prove our main theorem, that we can separate poly-size PRSs and log-size PRSs.
Formally, we will be proving the following theorem.
\begin{theorem}
  \label{thm:main}
  There exists a quantum oracle \(\mathcal{O}\) such that relative to \(\mathcal{O}\), \(\lambda\)-PRSs exist, but
  $(c\log \lambda)$-PRSs with $c>12$ do not exist.
\end{theorem}

The oracle necessary for the separation is actually the same oracle that Kretschmer used to separate PRSs and OWFs (\Cref{thm:kre21}).
It turns out that this oracle is also separating PD-OWFs from PRSs, we prove here that if~$\promiseBQP=\promiseQMA$, then we do not have PD-OWFs. This in addition with Barhoush and Salvail's result that short-PRSs are enough to build PD-OWFs (\Cref{thm:OWF}) will give us the theorem.

According to the above considerations, the main theorem follows directly from the following proposition.
\begin{proposition}
    \label{pro:owf}
    If PD-OWFs exist, then~$\promiseBQP\neq\promiseQMA$.   
\end{proposition}
\begin{proof}
Let $\lambda\in\mathbb{N}$ and $F:\dom\to\{0,1\}^{\ell(\lambda)}$ be a PD-OWF. Let us define a promise language $\lang = \lang_\textrm{yes}\cup\lang_\textrm{no}\cup\lang_\bot$ with $\lang\subseteq\{0,1\}^*$ where \emph{yes} instances have a pre-image with respect to $F_\lambda$ but \emph{no} instances do not. Formally,
\begin{multline}
  \lang_\textrm{yes} = \left\{(1^{\ell(\lambda)},x',y)\in \{1^{\ell(\lambda)}\}\times\dom\times\{0,1\}^{\ell(\lambda)} \middle\rvert\right. \\ \left.\exists x\in \dom, x' \prec x \text{ and } \Pr[y=F_{\lambda}(x)]\geq 1-\negl[\lambda] \right\}, \label{eq:langyes}
\end{multline}
\begin{multline}
  \lang_\textrm{no} = \left\{(1^{\ell(\lambda)},x',y)\in \{1^{\ell(\lambda)}\}\times\dom\times\{0,1\}^{\ell(\lambda)} \middle\rvert \right.\\ \left.\forall x\in \dom, x' \nprec x \text{ or } \Pr[y=F_{\lambda}(x)]\leq 1-\frac{1}{\lambda} \right\}. \label{eq:langno}  
\end{multline}
We claim that $\lang\in\promiseQMA$ but $\lang\not\in\promiseBQP$, thus there must be a separation between both complexity classes.
These claims are proven in~\Cref{lemma:linqma} and~\Cref{lemma:lnotinbqp} respectively.
\end{proof}

We now prove the two claims from the proposition. We start by showing that the language defined in~\cref{pro:owf} is in $\promiseQMA$, this is, we will construct an algorithm (verifier) that given an element of the domain and a (quantum) proof can distinguish if the element is a yes or no instance of the language.

\begin{lemma}
\label{lemma:linqma}
  Let \(F\) be a PD-OWF and let $\lang$ be the language defined in~\Cref{pro:owf}, then $\lang \in \promiseQMA$.
\end{lemma}
\begin{proof}
      We define a quantum polynomial-time algorithm \(\adv\) that given an element of the domain (\(1^{\ell(\lambda)},x',y) \in \{1^{\ell(\lambda)}\}\times\dom\times\{0,1\}^{\ell(\lambda)}\) and a classical proof
      $x\in \{0,1\}^{\ell(\lambda)}$, will check if the proof $x$ is indeed a pre-image of the PD-OWF by checking if it coincides with the input $y$.
    \pcb[linenumbering, head={\textbf{Algorithm 1}: $\adv((1^{\ell(\lambda)},x',y),x)$}]{
    \pcif x' \nprec x: \\
    \t \pcreturn 0 \\
        \pcfor 1 \leq i \leq 2\lambda: \\ 
        \t \pcif F_{\lambda}(x) \neq y: \\
    \t \t \pcreturn 0 \\
    \pcreturn 1}
    Note that the algorithm runs in polynomial-time trivially because computing $F_\lambda$ is done efficiently by definition and we make $2\lambda$ calls to it. We now prove that the algorithm distinguishes between the yes/no instances.
    
    \noindent (i) Let \((1^{\ell(\lambda)},x',y) \in \lang_\textrm{yes}\).
    Then by definition there exists a proof $x\in \dom$ such that
    \begin{align*}
        x' \prec x \text{ and } \Pr[y=F_\lambda(x)]\geq 1-\negl[\lambda].
    \end{align*}
    Then the proof $x$ will be an element of the input $((1^{\ell(\lambda)},x',y),x)$ for which the algorithm $\adv$ will output $1$ with high probability because
    \begin{align*}
        \pr{\adv((1^{\ell(\lambda)},x',y),x)=1} &= \pr{\forall 1 \leq i \leq 2\lambda, \text{ the execution of } F_{\lambda}(x) \text{ outputs y}} \\
        &\geq \left(1-\negl[\lambda]\right)^{2\lambda} 
        \geq 2/3,
    \end{align*}
    which holds whenever $\lambda\geq6$. Indeed, recall that $\negl[\lambda]\leq1/\lambda^c$ for all $c>1$, thus in particular $\negl[\lambda]\leq1/\lambda^2$, hence
    \begin{align*}
        (1-\negl[\lambda])^{2\lambda}\geq\left(1-\frac{1}{\lambda^2}\right)^{2\lambda}\geq\frac{2}{3},
    \end{align*}
    whenever $\lambda\geq 6$, where we used that we have an increasing function in $\lambda$.

    \noindent (ii) Let \((1^{\ell(\lambda)},x',y) \in \lang_\textrm{no}\). Then by definition for every potential proof $x\in\dom$ we have that either
    \begin{align*}
        x'\nprec x \quad\text{or}\quad\Pr[y=F_\lambda(x)]\leq 1-\frac{1}{\lambda}.
    \end{align*}
    Then for every possible input $((1^{\ell(\lambda)},x',y),x)$ the algorithm will output $0$ with high probability because
    \begin{align*}
        \pr{\adv((1^{\ell(\lambda)},x',y),x)=1} &= \pr{x'\nprec x\wedge\forall 1 \leq i \leq 2\lambda, \text{ the execution of } F_{\lambda}(x) \text{ outputs y}} \\
        &\leq \pr{\forall 1 \leq i \leq 2\lambda, \text{ the execution of } F_{\lambda}(x) \text{ outputs y}}\\
        &\leq \left(1-\frac{1}{\lambda}\right)^{2\lambda} 
         \leq e^{-2}
        \leq 1/3.
    \end{align*}
\end{proof}

\begin{lemma}
\label{lemma:lnotinbqp}
  Let \(F\) be a PD-OWF and let $\lang$ be the language defined in~\Cref{pro:owf}, then $\lang \notin \promiseBQP$.
\end{lemma}
\begin{proof}
  We will prove this by contradiction. Let us assume that instead \(\lang \in \promiseBQP\), this is, there exists a BQP algorithm \(\adv\) such that:
  \begin{enumerate}
    \item If \((1^{\ell(\lambda)},x',y) \in \lang_{\mathrm{yes}}\), then \(\pr{\adv(1^{\ell(\lambda)},x',y)=1} \geq 2/3\).
    \item If \((1^{\ell(\lambda)},x',y) \in \lang_\mathrm{no}\), then \(\pr{\adv(1^{\ell(\lambda)},x',y)=1} \leq 1/3\).
  \end{enumerate}
  Without loss of generality, we can assume that the algorithm \(\adv\) has completeness \(1-\frac{1}{\ell(\lambda)}\) and soundness \(\frac{1}{\ell(\lambda)}\).
  We will now show how we can construct a QPT algorithm \(\adv'\) that finds a pre-image of every \(F_\lambda\) with high probability when it exists, by querying the original BQP algorithm \(\adv\) at most \(\ell(\lambda)+1\) times.
  \pcb[linenumbering, head={\textbf{Algorithm 2}: $\adv'(1^{\ell(\lambda)},y)$}]{
    b \leftarrow \adv(1^{\ell(\lambda)},\epsilon,y) \\
    \pcif b=0: \\
    \t \pcreturn \bot \\
    x_0 \leftarrow \epsilon \\
    \pcfor 1 \leq i \leq \ell(\lambda):\\
    \t b \leftarrow \adv(1^{\ell(\lambda)},x_0 || 0,y)\\
    \t \pcif b=1:\\
    \t \t x_0 = x_0 || 0\\
    \t \pcelse:\\
    \t \t x_0 = x_0 || 1\\
    \pcreturn x_0
    }
  Indeed if $(1^{\ell(\lambda)},\epsilon,y)\in\lang_\mathrm{yes}$, then the probability that the algorithm $\adv '$ outputs a correct pre-image is very high
  \begin{equation}
  \label{eq:algo}
      \pr{y = \argmax_{y\in\{0,1\}^{\ell(\lambda)}}\pr{y=F_\lambda(x)}\given{x\leftarrow \adv'\left(1^{\ell(\lambda)},y\right)}}\geq \left(1-\frac{1}{\ell(\lambda)}\right)^{\ell(\lambda)+1}.
  \end{equation}
  However, this raises a contradiction with the security of the PD-OWF from the assumption~\Cref{def:QPD-OWF}, %
\begin{align*}
  &\pr[x\leftarrow \{0,1\}^{\ell(\lambda)}]{F_{\lambda}(\adv'(1^{\ell(\lambda)}, F_{\lambda}(x)))=F_{\lambda}(x)}\\
  &\hspace{2cm}= \pr[x\leftarrow \{0,1\}^{\ell(\lambda)}]{F_{\lambda}(\adv'(1^{\ell(\lambda)}, F_{\lambda}(x)))=F_{\lambda}(x) \given{x \in \mathcal{K}_{\lambda}}} \pr[x\leftarrow \{0,1\}^{\ell(\lambda)}]{x \in \mathcal{K}_{\lambda}} \\
  &\hspace{3cm}  + \pr[x\leftarrow \{0,1\}^{\ell(\lambda)}]{F_{\lambda}(\adv'(1^{\ell(\lambda)}, F_{\lambda}(x)))=F_{\lambda}(x) \given{x \notin \mathcal{K}_{\lambda}}} \pr[x\leftarrow \{0,1\}^{\ell(\lambda)}]{x \notin \mathcal{K}_{\lambda}} \\
  &\hspace{2cm}\geq \pr[x\leftarrow \{0,1\}^{\ell(\lambda)}]{F_{\lambda}(\adv'(1^{\ell(\lambda)}, F_{\lambda}(x)))=F_{\lambda}(x) \given{x \in \mathcal{K}_{\lambda}}} \pr[x\leftarrow \{0,1\}^{\ell(\lambda)}]{x \in \mathcal{K}_{\lambda}} \\
  &\hspace{2cm}\geq \pr[x\leftarrow \{0,1\}^{\ell(\lambda)}]{F_{\lambda}(\adv'(1^{\ell(\lambda)}, F_{\lambda}(x)))=F_{\lambda}(x) \given{x \in \mathcal{K}_{\lambda}}} \left(1 - \mu(\lambda)\right), \\
\end{align*}
where the first equality comes from the law of total probability and the second inequality comes from the property of \(\mathcal{K}_{\lambda}\).
We can rewrite the last element as:
\begin{align*}
  &\pr[x\leftarrow \{0,1\}^{\ell(\lambda)}]{y_3=y_1 \given{ y_1,y_2,y_3 \leftarrow F_{\lambda}(x), x_1 \leftarrow \adv'(1^{\ell(\lambda)}, y_2), x\in\mathcal{K}_{\lambda}}} \\
  &\hspace{1cm}\geq \pr[x\leftarrow \{0,1\}^{\ell(\lambda)}]{\left(y_3=y_2=y_1=\argmax_{y\in \{0,1\}^{\ell(\lambda)}}{\Pr[y=F_\lambda(x)]}\right) \land (x_1=x) \given{
  \begin{array}{c}
  y_1,y_2,y_3 \leftarrow F_{\lambda}(x)\\ x_1 \leftarrow \adv'(1^{\ell(\lambda)},y_2), x\in\mathcal{K}_{\lambda}
  \end{array}
  }}\\
  &\hspace{1cm}= \pr[x\leftarrow \{0,1\}^{\ell(\lambda)}]{F_{\lambda}(x)=\argmax_{y\in \{0,1\}^{\ell(\lambda)}}{\Pr[y=F_\lambda(x)]} \given{x \in \mathcal{K}_{\lambda}}}^3\\
  &\hspace{2cm}\cdot\pr[x\leftarrow \{0,1\}^{\ell(\lambda)}]{\adv'\left(1^{\ell(\lambda)},\argmax_{y\in \{0,1\}^{\ell(\lambda)}}{\Pr[y=F_\lambda(x)]}\right)=x \given{x \in \mathcal{K}_{\lambda}}}\\
  &\hspace{1cm}\geq (1 - \negl[\lambda])^{3} \left(1 - \frac{1}{\ell(\lambda)}\right)^{\ell(\lambda)+1},
\end{align*}
  where the first inequality comes from the law of total probability, and the last inequality from the definition of a PD-OWF and~\Cref{eq:algo}.
  This gives that:
  \begin{align*}
  \pr[x\leftarrow \{0,1\}^{\ell(\lambda)}]{F_{\lambda}(\adv'(1^{\ell(\lambda)},F_{\lambda}(x)))=F_{\lambda}(x)}& \geq (1 - \negl[\lambda])^{3} \left(1 - \frac{1}{\ell(\lambda)}\right)^{\ell(\lambda)+1}\left(1 - \mu(\lambda)\right)\\
  & \geq\left(1-O(\mu(\lambda))\right)\left(1 - \frac{1}{\ell(\lambda)}\right)^{\ell(\lambda)+1}.
  \end{align*}
  Note that this bound is not negligible since
  \begin{align*}
      \left(1 - \frac{1}{\ell(\lambda)}\right)^{\ell(\lambda)+1}
      \geq\frac{1}{10},
  \end{align*}
  whenever \(\lambda \geq 2\), which contradicts~\Cref{eq:sec}.
\end{proof}

\section*{Acknowledgements}
We deeply thank Alex Bredariol Grilo for many precious discussions. GM was supported by the Quantum Delta NL visitor's programme travel grant, which enabled the collaboration.

\ifsubmission 
\else
\fi
\bibliographystyle{alpha}
\renewcommand{\doi}[1]{\url{#1}}
\bibliography{./cryptobib/abbrev3,./cryptobib/crypto,references}
\end{document}